\documentclass[letterpaper, 10pt, conference,nofonttune]{ieeeconf}
\IEEEoverridecommandlockouts 
\usepackage{amsmath,amssymb}
\usepackage{graphicx,epstopdf}
\usepackage{euscript}
\newtheorem{theorem}{Theorem}[section]
\newtheorem{corollary}[theorem]{Corollary}
\newtheorem{lemma}[theorem]{Lemma}

\newtheorem{proposition}[theorem]{Proposition}

\newfont{\roma}{cmr10 scaled 1200}

\newcommand{\nline}  {{\mathbb N}}
\newcommand{\rline}  {{\mathbb R}}
\newcommand{\sline}  {{\mathbb S}}

\newcommand{\dd}   {{\rm d}\hbox{\hskip 0.5pt}}

\newcommand{\Ascr} {{\cal A}}

\newcommand{\Dscr} {{\cal D}}

\newcommand{\Iscr} {{\cal I}}

\newcommand{\Pscr} {{\cal P}}

\newcommand{\mm}    {{\hbox{\hskip 0.5pt}}}
\newcommand{\m}     {{\hbox{\hskip 1pt}}}

\newcommand{\bluff} {{\hbox{\raise 15pt \hbox{\mm}}}}
\newcommand{\sbluff}{{\hbox{\raise  7pt \hbox{\mm}}}}

\newcommand{\FORALL} {{\hbox{$\hskip 11mm \forall \;$}}}

\newcommand{\bbm}[1]{\left[\begin{matrix} #1 \end{matrix}\right]}

\title{\LARGE \bf {On the convergence of a numerical scheme for a boundary controlled 1D linear parabolic PIDE \vspace{-2mm}}}
\author{Soham Chatterjee and  Vivek Natarajan\vspace{-2mm}
\thanks{S. Chatterjee (soham.chatterjee@iitb.ac.in) and V. Natarajan (vivek.natarajan@iitb.ac.in) are with the Centre for Systems and Control, Indian Institute of Technology Bombay, Mumbai, India, 400076, Ph:+912225765385.}%
}

\begin{document}
\maketitle
\thispagestyle{empty}
\pagestyle{empty}

\begin{abstract}
We consider an 1D partial integro-differential equation (PIDE) comprising of an 1D parabolic partial differential equation (PDE) and a nonlocal integral term. The control input is applied on one of the boundaries of the PIDE. Partitioning the spatial interval into $n+1$ subintervals and approximating the spatial derivatives and the integral term with their finite-difference approximations and Riemann sum, respectively, we derive an $n^{\rm th}$-order semi-discrete approximation of the PIDE. The $n^{\rm th}$-order semi-discrete approximation of the PIDE is an $n^{\rm th}$-order ordinary differential equation (ODE) in time. We establish some of its salient properties and using them prove that the solution of the semi-discrete approximation converges to the solution of the PIDE as $n\to\infty$. We illustrate our convergence results using numerical examples. The results in this work are useful for establishing the null controllability of the PIDE considered. 
\end{abstract}



\section{Introduction} \label{sec1} 

The dynamics of many engineering processes evolve on an infinite-dimensional state-space and are often modelled using partial differential equations (PDEs) and partial integro-differential equations (PIDEs). Hence developing control algorithms for these models is a problem of practical interest. The early lumping approach to this problem is to construct an $n^{\rm th}$-order finite-dimensional approximation for the PDE/PIDE model, use the readily available ordinary differential equation (ODE) techniques to design a control signal for the approximate model, and finally show that as $n$ tends to infinity the control signal designed for the approximate model converges to a limiting control signal which solves the control problem for the PDE/PIDE model. This approach to the control of PDE models has been used to design stabilizing controllers in \cite{BaKu:1984}, \cite{BaIt:1991}, \cite{SiChNa:2022}, to design adaptive controllers and estimators in \cite{BaKr:2002}, \cite{BoBaKr:2003}, \cite{ChSuNa:2024}, \cite{ChSuNa:2025}, to study the controllability and observability in \cite{LeoZua::2002}, \cite{CiMiRo:2017}, \cite{LiuGuo::2021}, \cite{AyOzWa:2023} and to solve motion planning problems in \cite{UtMeKu:2007}, \cite{UtMeKu:2010},  \cite{ChNa:2020}, \cite{ChBaNa:2025}.

Recently in \cite{ChNa:2025}, using the early lumping approach, we solved a motion planning problem for an 1D parabolic PDE with discontinuous spatially-varying coefficients. In that work, we first obtained an $n^{\rm th}$-order semi-discrete approximation of the PDE by  partitioning the spatial interval into $n+1$ subintervals and replacing the spatial derivatives in the PDE with their finite-difference approximations. Then we constructed a control input which solved an appropriate motion planning problem for the $n^{\rm th}$-order semi-discrete approximation using the flatness technique. Finally, we showed that as $n$ tends to infinity the control input designed for the $n^{\rm th}$-order approximation converges to a limiting control input which solves the motion planning problem for the PDE. Our work yielded a new constructive proof for the null controllability of 1D parabolic PDEs with discontinuous spatially-varying coefficients.

In this work, we consider an 1D PIDE comprising of an 1D parabolic PDE with discontinuous spatially-varying coefficents and a nonlocal integral term. The control input is applied on one of the boundaries of the PIDE. Preliminary investigations suggest that the semi-discretization approach in \cite{ChNa:2025} can be adapted to establish the null controllability of the parabolic PIDE considered in this work. The approach in \cite{ChNa:2025} involves, in a crucial manner, establishing certain salient properties of the semi-discrete approximation of the parabolic PDE, see \cite[Sections III and VI]{ChNa:2025}. In this paper, we establish similar properties for the semi-discrete approximation of the parabolic PIDE of interest. Building on this, we plan to establish the null controllability of the parabolic PIDE (by adapting the techniques from \cite{ChNa:2025}) in a future work. We remark that the results in this work can potentially be used to solve other control problems for the 1D parabolic PIDE using the early lumping approach.


Briefly, the technical results in this paper are as follows.
We construct an $n^{\rm th}$-order semi-discrete approximation of the parabolic PIDE by partitioning the spatial interval into $n+1$ subintervals and approximating the spatial derivatives and the integral term with their finite-difference approximations and Riemann sum, respectively. 
We then derive an accuracy estimate for the discretization scheme, establish the uniform in $n$ analyticity of the semi-discrete system and prove a certain Sobolev type inequality. Using these results, we prove that for any continuous initial state and any input which is thrice differentiable, the solution of the $n^{\rm th}$-order semi-discrete approximation converges in the $L^2$ norm to the solution of the PIDE as $n\to\infty$. We also show that the convergence occurs in a stronger norm when additional smoothness assumptions are imposed on the initial state and input.

The rest of this paper is organized as follows: We introduce the 1D parabolic PIDE of interest in Section \ref{sec2} and prove the regularity of its solutions. In Section \ref{sec3}, we derive an $n^{\rm th}$-order semi-discrete approximation for the 1D parabolic PIDE and establish some of its salient properties. We present the main convergence results of this paper in Section \ref{sec4}. Finally, in Section \ref{sec5} we illustrate the convergence results of Section \ref{sec4} via numerical examples.

\noindent
{\bf Notations}. A function $\psi:[0,1]\to\rline$ belongs to  $PC^k[0,1]$ if the following holds: there exists a finite partitioning of $[0,1]$ into disjoint intervals such that on each of these intervals $\psi$ is the restriction of some $C^k[0,1]$ function to that interval. For $\psi\in PC^k[0,1]$, we let $\|\psi\|_{PC^k[0,1]}$ to be the supremum of $\|\psi\|_{C^k[a,b]}$ over all the intervals $[a,b]\subset[0,1]$ on which $\psi$ is $k$-times continuously differentiable. We write $PC^0[0,1]$ as $PC[0,1]$. Let $H^2(0,1)$ be the usual Sobolev space of order $2$. A function $\psi\in C^1[0,1]\cap H^2(0,1)$ is in $PC^{(2),1}[0,1]$ if $\psi_{xx}\in PC^1[0,1]$. For $\psi\in PC^{(2),1}[0,1]$, we let $\|\psi\|_{PC^{(2),1}[0,1]}=  \|\psi\|_{C^1[0,1]}+ \|\psi_{xx}\|_{PC^1[0,1]}$.

For $v= [v_1 \ v_2 \ \cdots \ v_n]^\top\in\rline^n$, let $\|v\|_p=\big(\sum_{i=1}^n |v_i|^p\big)^\frac{1}{p}$ for $p\geq1$,
$\|v\|_\infty= \max_{1\leq j \leq n} |v_j|$ and $\|v\|_{2d}=\sqrt{h}\, \|v\|_2$, where $h=1/(n+1)$. For $A\in\rline^{n\times n}$, we let $\|A\|_{2d}=\sup_{\|v\|_{2d}=1}\|A v\|_{2d}$. The extension operator $S_n:\rline^n \to L^\infty(0,1)$ is defined as follows: $[S_n v](0)=v_1$, $[S_n v](x)=v_j$ for $j\in \{1,2,\ldots n\}$ and $(j-1)h < x \leq j h$ and $[S_n v](x)=0$ for $nh < x\leq 1$. For a function $z\in PC[0,1]$, we let $R_n z = [z(h) \ z(2h) \ \cdots \ z(nh)]^\top$. \vspace{-1mm}

\section{The parabolic PIDE and its solution}\label{sec2} 

In this paper, we consider the following 1D parabolic PIDE on the interval $x\in[0,1]$ and $t\geq0$: \vspace{-1mm}
\begin{align}
 & u_t(x,t) = \theta(x) u_{xx}(x,t) + \sigma(x) u_x(x,t) \nonumber\\
 & \quad\hspace{15mm} + \lambda(x)u(x,t) + \int_0^x \varphi(x,\tilde x) u(\tilde x,t) \dd \tilde x\, , \label{eq:heat1} \\
 & \alpha_0 u_x(0,t) + \beta_0 u(0,t)=0, \ \  \alpha_1 u_x(1,t) + \beta_1 u(1,t)= f(t), \label{eq:heat2}\\[-4.2ex]\nonumber
\end{align}
Here the coefficients $\theta$, $\sigma$ and $\lambda$ are in $PC^1[0,1]$ with $\inf_{x\in[0,1]}\theta(x)>0$, the kernel $\varphi\in C^1([0,1];C^1[0,1])$ and $f$ is the boundary input. We let $L^2(0,1)$ be the state-space for the parabolic PIDE \eqref{eq:heat1}-\eqref{eq:heat2}. Then the state operator associated with it is $\Pscr:\Dscr(\Pscr)\subset L^2(0,1)\to L^2(0,1)$, where $\Dscr(\Pscr) =\{ w \in H^2(0,1) \m\big|\m \alpha_0 w_x(0) + \beta_0 w(0) = \alpha_1 w_x(1) + \beta_1 w(1)=0 \}$ \vspace{-1.5mm} and
\begin{align}
 &\Pscr w(x) = \theta(x) w_{xx}(x) + \sigma(x) w_x(x) + \lambda(x) w(x) \nonumber\\
 &\hspace{15mm}  + \int_0^x \varphi(x,\tilde x) w(\tilde x) \dd \tilde x \FORALL w\in \Dscr(\Pscr). \label{eq:Pscr}\\[-4.2ex]\nonumber
\end{align}
The graph norm on $\Dscr(\Pscr)$ is $\|w\|_{\Dscr(\Pscr)}= \|\Pscr w\|_{L^2(0,1)}+\|w\|_{L^2(0,1)}$. For any $w\in H^2(0,1)$ we let $\Pscr w$ be the function obtained by applying $\Pscr$ (as an integro-differential operator) to $w$, even if $w\notin \Dscr(\Pscr)$. Consider the operator $\Ascr:\Dscr(\Pscr)\subset L^2(0,1)\to L^2(0,1)$ with \vspace{-1mm}
\begin{equation}\label{eq:Ascr}
\Ascr w(x) = \theta(x) w_{xx}(x) + \sigma(x) w_x(x) + \lambda(x) w(x) \vspace{-1mm}
\end{equation}
for all $w \in \Dscr(\Pscr)$. Clearly $\Pscr=\Ascr+\Phi$, where $\Phi$ is a bounded operator from $L^2(0,1)$ to itself given by \vspace{-1mm}
$$\Phi w(x)=\int_0^x \varphi(x,\tilde x)w(\tilde x)\dd \tilde x \FORALL x\in [0,1].\vspace{-1mm}$$
Since $\Ascr$ generates an analytic semigroup on $L^2(0,1)$, see the discussion above Eq. (3) in \cite{ChNa:2025}, and $\Phi$ is a bounded operator, it follows from the perturbation result \cite[Ch. III, Thm. 2.10]{EnNa:2006} that $\Pscr$ generates an analytic semigroup $\sline$ on $L^2(0,1)$. We write $u(\cdot,t)$ as $u(t)$. For any initial state $u_0\in L^2(0,1)$ and input $f\in C^3[0,T]$, the solution
of \eqref{eq:heat1}-\eqref{eq:heat2} on the interval $[0,T]$ is the $L^2(0,1)$-valued continuous function defined \vspace{-1.5mm} as
\begin{align}
 &u(t) = \sline_t \left[u_0 - \nu f(0) \right] + \int_0^t \sline_{t-\tau}\big[ \Pscr \nu f(\tau)- \nu \dot{f}(\tau)\big]\dd \tau \nonumber\\[-1ex]
 &\hspace{55mm} + \nu f(t). \label{eq:mildsoln}\\[-5.5ex]\nonumber
\end{align}
Here $\nu(x)=\mu_1x^{\mu_2}$ for all $x\in [0,1]$ with $\mu_1,\mu_2\in \rline$ chosen such that $\mu_2\geq 3$ and $\mu_1(\mu_2\alpha_1+\beta_1)=1$. The function $\Pscr \nu$ is obtained by applying $\Pscr$ (as an integro-differential operator) to $\nu$, even though $\nu\notin\Dscr(\Pscr)$. The next lemma is a regularity result which shows that, under certain assumptions, $u$ in \eqref{eq:mildsoln} satisfies the parabolic PIDE \eqref{eq:heat1}-\eqref{eq:heat2} pointwise. It is an analog of the result for the parabolic PDE in \cite[Lemma 3.1]{ChNa:2025}.

\begin{lemma}\label{lm:reg}
Suppose that $u_0\in L^2(0,1)$ and $f\in C^3[0,T]$. Then $u\in C([0,T];L^2(0,1))$ defined in \eqref{eq:mildsoln} is the unique function in $C((0,T];PC^{(2),1}[0,1])\cap C^1((0,T];C^1[0,1])$ which satisfies \eqref{eq:heat1}-\eqref{eq:heat2} for $t\in(0,T]$ and $u(\cdot,0)=u_0$.
\end{lemma}
\begin{proof}
Let $w(x,t)= u(x,t) - \nu(x) f(t)$. Then a formal calculation using \eqref{eq:heat1}-\eqref{eq:heat2} gives that $w$ satisfies \vspace{-1mm}
\begin{align}
 &w_t(x,t) = \theta(x) w_{xx}(x,t) + \sigma(x) w_x(x,t) + \lambda(x) w(x,t) \nonumber\\[-0.5ex]
 & \quad + \int_0^x \varphi(x,\tilde x) w(\tilde x,t) \dd \tilde x + \Pscr \nu(x) f(t) - \nu(x) \dot f(t), \label{eq:wheat1} \\[-0.5ex]
 &\alpha_0 w_x(0,t) + \beta_0 w(0,t)=0,  \quad \alpha_1 w_x(1,t) + \beta_1 w(1,t)= 0. \label{eq:wheat2} \\[-5ex]\nonumber
\end{align}
The mild solution of \eqref{eq:wheat1}-\eqref{eq:wheat2} is \vspace{-1mm}
\begin{equation} \label{eq:msw}
 w(t)= \sline_t w_0 + \int_0^t \sline_{t-\tau} [\Pscr \nu f(\tau) - \nu \dot{f}(\tau)]\dd \tau, \vspace{-1mm}
\end{equation}
where $w_0=u_0-\nu f(0)\in L^2(0,1)$. Using \cite[Ch. 4, Corollary 2.5 and Thm 3.5 (i)]{Pazy:1983} we get that $\Pscr w, \dot w \in C((0,T];L^2(0,1))$ and $w$ satisfies \eqref{eq:wheat1}-\eqref{eq:wheat2} for $t>0$. Changing the variable from $\tau$ to $s=t-\tau$ in the integral in \eqref{eq:msw}, then differentiating \eqref{eq:msw} and then changing the variable back from $s$ to $\tau = t-s$ we obtain \vspace{-1mm} that
\begin{align*}
&\dot w(t) = \Pscr \sline_t w_0 + \sline_t(\Pscr\nu f(0)-\nu\dot f(0))\\
&\hspace{30mm} + \int_0^t \sline_{t-\tau} [\Pscr \nu \dot f(\tau) - \nu \ddot{f}(\tau)]\dd \tau \\[-4.5ex]
\end{align*}
for $t>0$. Since $\Pscr^2\sline_t w_0 \in C((0,T];L^2(0,1))$ (because $\sline$ is an analytic semigroup), again using \cite[Ch. 4, Thm. 3.5 (i)]{Pazy:1983} to deduce the regularity of the remaining terms on the right-side of the above expression, we get that $\Pscr \dot w \in C((0,T];L^2(0,1))$. So $\dot w \in C((0,T]; C^1[0,1])$. Using this and $w\in C((0,T];\Dscr(\Pscr))$ shown earlier, it follows that $\sigma w_x\in C((0,T];PC[0,1])$ and all  terms in \eqref{eq:wheat1} except $\sigma w_x$ and $\theta w_{xx}$ are in $C((0,T];PC^1[0,1])$. Hence $\theta w_{xx}\in C((0,T];PC[0,1])$, which implies that $\sigma w_x\in C((0,T]; PC^1[0,1])$. So all terms in \eqref{eq:wheat1} except $\theta w_{xx}$ are in $C((0,T];PC^1[0,1])$. Consequently $\theta w_{xx}\in C((0,T];PC^1[0,1])$ and so $w\in C((0,T]; PC^{(2),1}[0,1])$.

From the above discussion, in particular $u(x,t)= w(x,t) + \nu(x) f(t)$ and \eqref{eq:msw}, we get that $u$ defined in \eqref{eq:mildsoln} is in $C((0,T];PC^{(2),1}[0,1])\cap C^1((0,T];C^1[0,1])$ and satisfies \eqref{eq:heat1}-\eqref{eq:heat2} for $t>0$. The uniqueness of $u$ follows easily from \cite[Ch. 4, Corollary 3.3]{Pazy:1983}.
\end{proof}
\vspace{-1mm}

\section{Semi-discrete approximation}\label{sec3} 

A semi-discrete approximation of \eqref{eq:heat1}-\eqref{eq:heat2} was obtained in \cite{ChNa:2025} in the absence of the integral term by approximating the spatial derivatives in \eqref{eq:heat1}-\eqref{eq:heat2} using their finite-differences. Using it and by approximating the integral term with its Riemann sum we arrive at the following $n^{\rm th}$-order semi-discrete approximation of \eqref{eq:heat1}-\eqref{eq:heat2}: \vspace{-1mm}
\begin{equation} \label{eq:semi_disc}
 \dot v_n(t) = P_n v_n (t) + B_n f_n(t) \!\! \FORALL t\in[0,T]. \vspace{-2mm}
\end{equation}
Here $v_n(t)\in\rline^n$ and $P_n\in \rline^{n\times n}$ and $B_n\in \rline^{n\times 1}$ are defined as follows: $P_n = \Theta_n L_n + \Sigma_n D_n + \Lambda_n + \Phi_n$ with $\Theta_n, L_n,\Sigma_n, D_n,\Lambda_n,\Phi_n\in\rline^{n\times n}$, \vspace{-1mm}
$$\Theta_n = \text{diag}[\theta(h) \ \ \theta(2h) \ \ \cdots \ \  \theta(nh)], \vspace{-1mm}$$
$$\Sigma_n = \text{diag}\bbm{\sigma(h) \ \ \sigma(2h) \ \ \cdots \ \  \sigma(nh)}, \vspace{-1mm}$$
$$\Lambda_n = \text{diag}\bbm{\lambda(h) \ \ \lambda(2h) \ \ \cdots \ \ \lambda(nh)}, \vspace{-1mm}$$
$${\small L_n\! =\! \frac{1}{h^2}\!
      \bbm{
        4 r_0-2 & 1-r_0 & 0 & 0 & \hspace{2mm}\cdots &0 \\
        1 & -2 & 1 & 0 &\hspace{2mm}\cdots &0 \\
        0 &  1 &-2 & 1 &  &0  \\
        \vdots  & \hspace{3.5mm}\ddots & \hspace{1mm} \ddots & \hspace{1mm} \ddots & &\vdots \\
        0 &  & \hspace{-8mm}1 & \hspace{-8.5mm}-2 & \hspace{-8.5mm}1 &0 \\
        0 &  \hspace{-8mm}\cdots &\hspace{-8.5mm}0 & \hspace{-8.5mm}1  &\hspace{-8.5mm}-2 &1\\
        0 & \hspace{-8mm}\cdots &\hspace{-8.5mm}0 &\hspace{-8.5mm}0 &\hspace{-7mm} 1-r_1 &4 r_1-2}, \vspace{-1mm}}$$
$${\small D_n =\! \frac{1}{h}\!
    \bbm{
      h q_0& 0 & \cdots & \cdots &\cdots &0 \\
      -1 & 1 & 0 &\cdots &\cdots &0 \\
      0 &  -1 &1 &0 &\cdots &0  \\
      \vdots  & \ddots & \ddots & \ddots & \ddots & \vdots \\
      0 & \cdots  &0 & -1 & 1 &0 \\
      0 & \cdots & \cdots &0  &-1 &1}, \ \
  B_n \!=\! \frac{1}{h^2}\!\bbm{0 \\ 0 \\ \vdots \\ 0 \\ 0 \\ b_n}, \vspace{-1mm}}$$
$${\small \Phi_n = h\bbm{\varphi(h,h) &0 &\cdots &0\\
      \varphi(2h,h) &\varphi(2h,2h) &\cdots &0\\
      \vdots  &\vdots  &\ddots &0\\
      \varphi(nh,h) &\varphi(nh,2h) &\cdots &\varphi(nh,nh)}.\vspace{-0.5mm}}$$
Here $h=1/(n+1)$, $r_0=\alpha_0/(3\alpha_0-2h\beta_0)$, $r_1=\alpha_1/(3\alpha_1+2h \beta_1)$, $q_0 =-\beta_0 / (\alpha_0 -h\beta_0)$ and $b_n = 2h\theta(nh)/(3\alpha_1 + 2h\beta_1)$. While $r_0$, $r_1$, $q_0$ and $b_n$ are well-defined for all $n\gg1$, we assume that they are well-defined for all $n$. This is reasonable since we are interested in the solutions of \eqref{eq:semi_disc} in the limit $n\to\infty$. A simple calculation gives \vspace{-1.5mm}
\begin{equation} \label{eq:PhiBD}
 \|\Phi_n\|_{2d} \leq \sup_{x,\tilde x\in [0,1]}|\varphi(x,\tilde x)| \FORALL n\geq 1. \vspace{-1.5mm}
\end{equation}

The state operator of the $n^{\rm th}$-order semi-discrete approximation of the parabolic PDE in \cite{ChNa:2025} is \vspace{-1.5mm}
\begin{equation}\label{eq:An}
 A_n = \Theta_n L_n + \Sigma_n D_n + \Lambda_n. \vspace{-1.5mm}
\end{equation}
Clearly, $P_n = A_n + \Phi_n$. In this section, we present three results about $P_n$: In Lemma \ref{lm:approxAq} we derive an estimate for the accuracy of the approximation $P_n$ of $\Pscr$. We establish the uniform in $n$ analyticity of $e^{P_n t}$ in Proposition \ref{pr:eigenA} and prove a certain Sobolev type inequality involving $P_n$ in Lemma \ref{lm:disc_sob}. We have derived analogous results for $A_n$ in \cite{ChNa:2025}. We use them along with the fact that $P_n$ and $A_n$ differ by a bounded perturbation to establish the results in this section.

Let $\Iscr$ be the collection of all the points where either $\theta$, $\sigma$ or $\lambda$ is not differentiable (so $\Iscr$ contains the points where these functions are discontinuous).

\begin{lemma}\label{lm:approxAq}
Fix $T>0$ and let $\rho\in[0,T)$. Consider $\xi \in C([\rho,T],PC^{(2),1}[0,1])$ satisfying $\alpha_0 \xi_x(0,t) + \beta_0 \xi(0,t)=0$. Define $f_\xi(t)=\alpha_1 \xi_x(1,t) + \beta_1 \xi(1,t)$. Let the points at which $\xi_{xx}(\cdot,t)$ is not differentiable be a subset of $\Iscr$ for all $t\in[\rho,T]$. Then, for all $n\gg 1$ there exists a $C$ independent of $\xi$ and $n$ such that
\vspace{-2mm}
\begin{align}
&\sup_{t\in[\rho,T]}\|R_n \Pscr \xi(\cdot,t) - P_n R_n \xi(\cdot,t) - B_n f_\xi(t)\|_{2d} \nonumber\\[-1ex]
  &\hspace{5mm}\leq\ C \sqrt{h} \sup_{t\in[\rho,T]} \|\xi(\cdot,t)\|_{PC^{(2),1}[0,1]} \label{eq:pwiseJ1one} \\[-4ex]\nonumber
\end{align}
\end{lemma}

\begin{proof}
Fix $n\gg 1$. Recall the operator $\Ascr$ from \eqref{eq:Ascr} and matrix $A_n$ from \eqref{eq:An}. Note that $\|v\|_{2d}=\sqrt{h}\|v\|_2$ for any $v\in\rline^n$. From \cite[Lemma 3.2]{ChNa:2025} we have \vspace{-1.2mm}
\begin{align}
  &\sup_{t\in[\rho,T]}\|R_n \Ascr \xi(\cdot,t) - A_n R_n \xi(\cdot,t) - B_n f_\xi(t)\|_{2d} \nonumber\\[-1ex]
  &\hspace{5mm}\leq C_1\sqrt{h} \sup_{t\in[\rho,T]} \|\xi(\cdot,t)\|_{PC^{(2),1}[0,1]} \label{eq:pwiseJ1oneA} \\[-4.7ex]\nonumber
\end{align}
for some $C_1>0$ independent of $n$ and $\xi$. For all $j\in \{1,2,\ldots n\}$ and $t\in [\rho,T]$ we also have
\begin{align}
  &\big|[R_n (\Pscr-\Ascr) \xi(\cdot,t)-(P_n-A_n)R_n \xi(\cdot,t)]_j\big|\nonumber\\[-0.5ex]
  =&\Big|\sum_{m=1}^{j}\int_{mh-h}^{mh} \big[\varphi(jh,\tilde x)\xi(\tilde x,t) - \varphi(jh,mh)\xi(mh,t)\big] \dd \tilde x \Big| \nonumber\\[-0.5ex]
  \leq& h \sup_{x\in[0,1]} \|\varphi(x,\cdot)\|_{C^1[0,1]} \sup_{t\in[\rho,T]} \|\xi(\cdot,t)\|_{C^1 [0,1]}. \label{eq:newest} \\[-4.5ex]\nonumber
\end{align}
Here we have used $[\cdot]_j$ to denote the $j^{\rm th}$ component of vectors in $\rline^n$. From the above estimate we can conclude \vspace{-1.2mm} that
\begin{align*}
 &\sup_{t\in [\rho,T]}\|[R_n (\Pscr-\Ascr) -(P_n-A_n)R_n] \xi(\cdot,t)\|_{2d} \nonumber\\[-1.5ex]
 &\hspace{10mm}\leq C_2\sqrt{h}\sup_{t\in [\rho,T]}\|\xi(\cdot,t)\|_{PC^{(2),1}[0,1]}\\[-5ex]
\end{align*}
for some $C_2>0$ independent of $n$ and $\xi$. The estimate \eqref{eq:pwiseJ1one} now follows from \eqref{eq:pwiseJ1oneA} and \eqref{eq:newest} via the triangle inequality.
\end{proof}


Next we show that the semigroup $e^{P_n t}$ has a uniform in $n$ growth bound and is also uniformly analytic in $n$. Given a $T>0$, from \cite[Prop. 6.1]{ChNa:2025} we get that  there exists a $C_T>0$ independent of $n$ and $k$ such \vspace{-1.2mm} that
\begin{equation}\label{eq:AnenT}
 \|A_n^k e^{A_n t}\|_{2d} \leq\!  \frac{C_T^{k+1} k!}{t^k} \vspace{-1.2mm}
\end{equation}
for all $t\in (0,T]$ and each $k\geq 0$ and $n\gg 1$. \vspace{1mm}
\begin{proposition}\label{pr:eigenA}
Consider the $n^{\rm th}$-order ODE \vspace{-1.2mm}
\begin{equation} \label{eq:lmODE}
 \dot v(t) =  P_n v(t) + \eta(t), \qquad v(0)=v_0, \vspace{-1.2mm}
\end{equation}
with $\eta\in L^\infty([0,\infty);\rline^n)$. Then, there exists $M,\omega>0$ independent of $n$ such that the solution $v$ to the ODE \vspace{-1mm} satisfies
\begin{equation}
 \|v(t)\|_{2d} \leq M \Big(\!e^{\omega t} \|v_0\|_{2d}+\int_0^t\!\!  e^{\omega (t-\tau)} \| \eta(\tau)\|_{2d}\ \dd\tau\!\Big) \ \ \forall \ t\geq 0. \label{eq:lmbnd} \vspace{-0.5mm}
\end{equation}
Furthermore, for any $T>0$ there exists an $M_T>0$ independent of $n$ and $t$ such that \vspace{-1.5mm}
\begin{equation}\label{eq:analytic}
 \|P_n^ke^{P_n t}\|_{2d} \leq \frac{M_T^{k+1}k!}{t^k} \vspace{-1mm}
\end{equation}
for all $t\in (0,T]$ and each $k\geq 0$ and $n\gg 1$.
\end{proposition}

\begin{proof}
From \cite[Prop. 3.3]{ChNa:2025} we have $M_a>0$ and $\omega_a>0$ independent of $n$ and $t$ such that $\|e^{A_n t}\|_{2d} \leq M_ae^{\omega_a t}$ for all $t\geq 0$ and each $n\gg 1$. Since $P_n = A_n + \Phi_n$ and the perturbation $\Phi_n$ bounded uniformly in $n$ (see \eqref{eq:PhiBD}), it follows from \cite[Ch. 3, Thm. 1.1]{Pazy:1983} and the growth bound for $e^{A_nt}$ that \vspace{-3mm}
\begin{equation}\label{eq:pgrowth}
 \|e^{P_n t}\|_{2d} \leq M e^{\omega t} \FORALL t\geq 0 \vspace{-1.1mm}
\end{equation}
with $M=M_a$ and $\omega = \omega_a+ M_a\|\Phi_n\|_{2d}$. Using this we can we can now bound the terms containing $e^{P_n t}$ and $e^{P_n(t-\tau)}$ in the solution $v(t) = e^{P_n t}v_0 + \int_0^t e^{P_n(t-\tau)} \eta(\tau)\dd \tau$ to get \eqref{eq:lmbnd}.

Fix $n \gg 1$ and $z\in \rline^n$. Let $v(t) = tP_ne^{P_n t}z$ for $t\in (0,T]$. Differentiating $v$ and using $P_n = A_n + \Phi_n$ gives us the ODE: $\dot v(t)=A_n v(t)+ \Phi_n v(t) + \frac{v(t)}{t}$ for $t\in (0,T]$ and $v(0)=0$. Expressing the solution of this ODE using the variation of constants formula gives us
\vspace{-5mm}
\begin{equation}\label{eq:sol1}
 v(t) = \int_0^t e^{A_n(t-\tau)}\Phi_nv(\tau)\, \dd \tau + \int_0^t e^{A_n(t-\tau)}\frac{v(\tau)}{\tau}\, \dd \tau \vspace{-1mm}
\end{equation}
for all $t\in (0,T]$. We can rewrite this equation \vspace{-1mm} as
\begin{align}
 v(t) =& \int_0^t e^{A_n(t-\tau)}\Phi_nv(\tau)\, \dd \tau + \int_{t/2}^t e^{A_n(t-\tau)}\frac{v(\tau)}{\tau}\, \dd \tau \nonumber\\[-1.2ex]
 & + \int_0^{t/2}e^{A_n(t-\tau)} P_n e^{P_n \tau}z\, \dd \tau \!\!\!\!\!\! \FORALL t\in [0,T].\label{eq:randm}\\[-4.6ex] \nonumber
\end{align}
We have replaced $v(\tau)/\tau$ with $P_n e^{P_n \tau}$ on the right side of \eqref{eq:sol1} to get the last integral. Using \eqref{eq:PhiBD} and \eqref{eq:AnenT} to bound the terms in the first two integral on the right side of \eqref{eq:randm} we \vspace{-1mm} get
\begin{align}
 &\left\| \int_0^t e^{A_n(t-\tau)}\Phi_nv(\tau)\, \dd \tau \right\|_{2d} \leq c_1\int_0^t\|v(\tau)\|_{2d}\, \dd \tau,\label{eq:estimate1} \\[0.5ex]
 &\left\| \int_{t/2}^t e^{A_n(t-\tau)}\frac{v(\tau)}{\tau}\, \dd \tau \right\|_{2d} \leq \frac{c_2}{t}\int_0^t \|v(s)\|_{2d}\, \dd s \label{eq:estimate2}\\[-4.5ex]\nonumber
\end{align}
%
for all $t\in (0,T]$ and some $c_1,c_2>0$ independent of $n$, $z$ and $t$.
In \eqref{eq:estimate2} we have also used $2\tau>t$. From \eqref{eq:PhiBD} and \eqref{eq:AnenT} it follows that $\|e^{A_n t}P_n\|_{2d} \leq C_a/t$ for all $t\in (0,T]$ and some $C_a$ independent of $n$.  Using this and \eqref{eq:pgrowth} we can bound the third integral term on the right side of \eqref{eq:randm} as \vspace{-1mm} follows
\begin{equation}
  \left\|\int_0^{t/2}\!\!\!\!\!\!e^{A_n(t-\tau)} P_n e^{P_n \tau}z\, \dd \tau \right\|_{2d} \!\!\!\leq \!\!\int_0^{t/2} \!\frac{c_3\|z\|_{2d}}{(t-\tau)}\, \dd \tau \leq c_3\|z\|_{2d} \vspace{-1mm} \label{eq:estimate3}
\end{equation}
for all $t\in (0,T]$ and some $c_3$ independent of $n$, $z$ and $t$. Now using \eqref{eq:estimate1}-\eqref{eq:estimate3} to bound the terms in \eqref{eq:randm} we get
$\|v(t)\|_{2d} \leq c_3\|z\|_{2d} + \frac{c_1 t+c_2}{t}\int_0^t  \|v(\tau)\|_{2d}\, \dd \tau$ for all $t\in (0,T]$. Applying the Gronwall's lemma then gives us $\|v(t)\|_{2d} \leq N_T\|z\|_{2d}$ for all $t\in (0,T]$ with $N_T = c_3e^{c_1 T+ c_2}$. Since $v(t) = tP_ne^{P_n t}z$, we therefore have $\|P_n e^{P_n t}\|_{2d} \leq N_T/t$ for all $t\in (0,T]$ and any $n\gg 1$. From this and the facts that $P_n^k e^{P_n t}\! =\! (P_n e^{P_n \frac{t}{k}})^k$ and $k^k\leq k!e^k$ we get \eqref{eq:analytic} with $M_T=e N_T$.
\end{proof}

Finally, we prove a discrete Sobolev type inequality.
\begin{lemma}\label{lm:disc_sob}
For any $v\in \rline^n$ there exists an $M>0$ independent of $n$ and $v$ such that \vspace{-1mm}
\begin{equation}\label{eq:sobemb}
 \|v\|_{\infty} + \|D_n v\|_{\infty} \leq M\left(\|v\|_{2d} + \|P_n v\|_{2d}\right)\vspace{1mm}
\end{equation}
\end{lemma}

\begin{proof}
From \cite[Eq. (58)]{ChNa:2025} we get the inequality \vspace{-1mm}
\begin{equation}\label{eq:Anest}
 \!\!\!\|v\|_{\infty} \!+\! \|D_n v\|_{\infty} \!\leq\! M_1 \left(\|v\|_{2d} + \|A_n v\|_{2d}\right) \quad \forall \, v\in \rline^n, \vspace{-1mm}
\end{equation}
where $M_1>0$ is independent of $n$ and $v$. Since $A_n = P_n - \Phi_n$, using \eqref{eq:PhiBD} we get $\|A_n v\|_{2d} \leq C(\|P_nv\|_{2d} + \|v\|_{2d})$ for some $C$ independent of $n$. Using this in \eqref{eq:Anest} gives us \eqref{eq:sobemb}.
\end{proof}
\vspace{-0.5mm}

\section{Convergence results} \label{sec4} 

We now present two results which establish and characterize the convergence of the solution of the $n^{\rm th}$-order semi-discrete approximation \eqref{eq:semi_disc} to the solution of the PIDE \eqref{eq:heat1}-\eqref{eq:heat2} as $n\to\infty$. In the first result given below, we consider the initial states of \eqref{eq:heat1}-\eqref{eq:heat2} to be continuous functions. Recall the operator $S_n:\rline^n\to L^\infty(0,1)$ from the notations at the end of Section \ref{sec1}. In the proof we will often use the fact that $\|S_n v\|_{L^2(0,1)}=\|v\|_{2d}$ for each $v\in\rline^n$. \vspace{1mm}

\begin{theorem}\label{th:semidisc_conv}
Fix $T\!>0$. Consider $u_0\in  C[0,1]$, a function $f\in C^3[0,T]$ and a sequence $\{f_n\}_{n=1}^\infty$ in $C^3[0,T]$ such that \vspace{-1mm}
\begin{equation} \label{eq:fn_limit}
 \lim_{n\to\infty}\|f_n-f\|_{C^1[0,T]}=0. \vspace{-1.5mm}
\end{equation}
Let $u$ be the solution of \eqref{eq:heat1}-\eqref{eq:heat2} on the time interval $[0,T]$ with

\m\vspace{-5mm}\noindent  initial state $u_0$ and input $f$. Let $v_n$ be the solution of the $n^{\rm th}$-order semi-discrete system \eqref{eq:semi_disc} on the time interval $[0,T]$ with initial state $R_n u_0$ and input $f_n$. Then we have \vspace{-1mm}
\begin{equation}\label{eq:semdisconv2}
  \lim_{n\to\infty}\sup_{t\in[0,T]}\|u(\cdot,t) - S_n v_n(t)\|_{L^2(0,1)}= 0. \vspace{1mm}
\end{equation}
\end{theorem}
\begin{proof}
From Lemma \ref{lm:reg} we get $u\!\in\! C((0,T];PC^{(2),1}[0,1])$ \! \!\!and $\dot u\in C((0,T];C^1[0,1])$. Using this in \eqref{eq:heat1}
it is easy to see that the points at which $u_{xx}(\cdot,t)$ is not differentiable is a subset of $\Iscr$ for all $t\in(0,T]$. First we will show that \vspace{-2.2mm}
\begin{equation} \label{eq:initcont}
 \lim_{t\to 0} \sup_{n\in\nline} \|v_n(t)-R_n u_0\|_{2d}=0. \vspace{-1.6mm}
\end{equation}
Let $e_n(t) = R_n u(\cdot,t) - v_n(t) + R_n \nu (f_n(t)-f(t))$. Then \vspace{-1.6mm}
\begin{align}
 &u(\cdot,t)-S_n v_n(t) = [u(\cdot,t)-S_n R_n u(\cdot,t)]\nonumber\\[-0.5ex]
 &\hspace{20mm} + S_n e_n(t) - [S_nR_n\nu(f_n(t)-f(t))]. \label{eq:useexp}\\[-4.5ex]\nonumber
\end{align}
For any $\rho\in(0,T)$ we will show that \vspace{-1.3mm}
\begin{equation} \label{eq:enl2}
 \lim_{n\to\infty}\sup_{t\in[\rho,T]}\|e_n(t)-e^{P_n (t-\rho)} e_n(\rho) \|_{2d} =0. \vspace{-1.3mm}
\end{equation}
The equations \eqref{eq:initcont} and \eqref{eq:enl2} imply \eqref{eq:semdisconv2}. Indeed, fix $\epsilon>0$. Using $u\in C([0,T]; L^2(0,1))$ and \eqref{eq:initcont} choose $\rho_\epsilon$ such that $\|u(\cdot,t)-u_0\|_{L^2(0,1)}\leq \epsilon$ and $\|S_n v_n(t)-S_nR_n u_0\|_{L^2(0,1)}<\epsilon$ for all $t\in[0,\rho_\epsilon]$ and $n\in\nline$. Since $u_0\in C[0,1]$, there exists $\tilde n_\epsilon$ such that $\|u_0 - S_n R_n u_0\|_{L^2(0,1)}<\epsilon$ for $n\geq\tilde n_\epsilon$. From these inequalities we get \vspace{-1.5mm}
\begin{equation} \label{eq:pfrho}
 \sup_{t\in[0, \rho_\epsilon]} \|u(\cdot,t)-S_n v_n(t)\|_{L^2(0,1)} <3\epsilon \vspace{-1.5mm}
\end{equation}
for all $n\geq\tilde n_\epsilon$. Next choose $n_\epsilon\!>\tilde n_\epsilon$ such that for each \vspace{-1mm} $n>n_\epsilon$,
\begin{align}
 &\sup_{t\in[0,T]}\|S_nR_n\nu(f_n(t)-f(t))\|_{L^2} \nonumber\\[-2ex]
 & \hspace{10mm} + \sup_{t\in[\rho_\epsilon, T]} \|u(\cdot,t) - S_n R_n u(\cdot,t)\|_{L^2}<2\epsilon, \label{eq:snrn}\\[-0.5ex]
 &\sup_{t\in[\rho_\epsilon,T]} \|e_n(t)-e^{P_n (t-\rho_\epsilon)}e_n(\rho_\epsilon)\|_{2d} <\epsilon. \label{eq:ent-} \\[-5ex]\nonumber
\end{align}
The existence of such an $n_\epsilon$ is guaranteed by  \eqref{eq:fn_limit}, the uniformly continuity of $u$ on $[0,1]\times[\rho_\epsilon,T]$ (note that $u\in C([\rho_\epsilon,T];C[0,1])$) and \eqref{eq:enl2}. Then for $n\geq n_\epsilon$, using \eqref{eq:pfrho} and \eqref{eq:snrn} in \eqref{eq:useexp} with $t=\rho_\epsilon$ we get $\|e_n(\rho_\epsilon)\|_{2d} <5\epsilon$. From this, via Proposition \ref{pr:eigenA}, we get $\|e^{P_n (t-\rho_\epsilon)} e_n(\rho_\epsilon) \|_{2d}\leq 5\epsilon M e^{\omega T}$ for $t\in[\rho_\epsilon,T]$. This estimate and \eqref{eq:ent-} give $\|S_n e_n(t)\|_{L^2(0,1)} <\epsilon+5\epsilon M e^{\omega T}$ for $t\in[\rho_\epsilon,T]$. Using this and \eqref{eq:snrn} in \eqref{eq:useexp} we get $\sup_{t\in[\rho_\epsilon,T]} \|u(\cdot,t)-S_n v_n(t)\|_{L^2(0,1)} <3\epsilon + 5\epsilon Me^{\omega T}$ for all $n\geq n_\epsilon$. Since \eqref{eq:pfrho} and the above estimate can be established for any $\epsilon>0$, it follows that \eqref{eq:semdisconv2} holds. We will now complete the proof of the theorem by showing that \eqref{eq:initcont} and \eqref{eq:enl2} hold.

We will first prove \eqref{eq:initcont}. Recall $\nu$ defined below \eqref{eq:mildsoln}. Fix $\epsilon>0$ and $\tilde u_0\in C^3[0,1]$ such that $\tilde u_0$ has a compact support in $(0,1)$ and $\|u_0 - \nu f(0) - \tilde u_0\|_{L^2(0,1)}< \epsilon$. Define $z_n(t)=v_n(t)-R_n\nu f_n(t)-R_n \tilde u_0$. Then  \vspace{-1.5mm}
\begin{equation} \label{eq:znode}
 \dot z_n(t) = P_n z_n(t) +  \eta_n(t), \quad z_n(0)=R_n (u_0 - \nu f_n(0) - \tilde u_0),  \vspace{0mm}
\end{equation}
where $\eta_n(t)= (B_n+P_n R_n \nu) f_n(t) + P_n R_n \tilde u_0 - R_n\nu \dot f_n(t)$. Since $u_0 - \nu f(0) - \tilde u_0 \in C[0,1]$ and \eqref{eq:fn_limit} holds, there exists $\tilde t_\epsilon>0$ and $n_\epsilon$ such that for $t\leq \tilde t_\epsilon$ and $n\geq n_\epsilon$, \vspace{-1.1mm}
\begin{equation} \label{eq:trbea}
\|R_n\nu(f_n(t)-f_n(0))\|_{2d}\leq \epsilon, \quad \|z_n(0)\|_{2d}<2\epsilon. \vspace{-1mm}
\end{equation}
Applying Lemma \ref{lm:approxAq} with $\xi(\cdot,t)=\tilde u_0$ and $\xi(\cdot,t)=\nu$ we get, for some $C_1>0$ and all $n\gg1$, that $\|P_n R_n \tilde u_0\|_{2d} \leq \|R_n\Pscr\tilde u_0\|_{2d} + C_1\sqrt{h}$ \vspace{-1mm} and
\begin{equation}  \label{eq:nulemma}
 \|P_n R_n \nu + B_n\|_{2d} \leq \|R_n\Pscr\nu\|_{2d} + C_1\sqrt{h}. \vspace{-1mm}
\end{equation}
Using this, the fact that $\|R_n\nu\|_{2d}$,\! $\|R_n\Pscr\nu\|_{2d}$ and $\|R_n\Pscr\tilde u_0\|_{2d}$ can each be bounded by $C_2>0$, and \eqref{eq:fn_limit}, it follows that $\|\eta_n(t)\|_{2d}< C_3 \sqrt{h}$ for some $C_3>0$, each $t\in[0,T]$ and all $n$. So applying Proposition \ref{pr:eigenA} to \eqref{eq:znode} we can get \vspace{-1mm}
$$ \|z_n(t)\|_{2d} \leq M e^{\omega t} \|z_n(0)\|_{2d} + C t \FORALL t\in[0,T]  \vspace{-1mm} $$
and some $C>0$ that depends on $\tilde u_0$ and so on $\epsilon$. The above inequality and \eqref{eq:trbea} guarantee the existence of $t_\epsilon<\tilde t_\epsilon$ such that $\|z_n(t)\|_{2d}<4M\epsilon$ for $n\geq n_\epsilon$ and $t\in[0,t_\epsilon]$. From this, the expression $v_n(t)-R_n u_0 = z_n(t) - z_n(0) + R_n\nu(f_n(t)-f_n(0))$ and \eqref{eq:trbea}, it now follows easily that
$\|v_n(t)-R_n u_0\|_{2d}<7M\epsilon$ for $n\geq n_\epsilon$ and $t\in[0,t_\epsilon]$. Finally, using $\lim_{t\to0} \|v_n(t)-R_n u_0\|_{2d}=0$ for each $n$ and redefining $t_\epsilon$ if needed, we get $\|v_n(t)-R_n u_0\|_{2d}<7M\epsilon$ for all $n\in\nline$ and $t\in[0,t_\epsilon]$. Since this estimate can be established for any $\epsilon>0$, we conclude that \eqref{eq:initcont} holds.

We will now complete this proof by establishing \eqref{eq:enl2}. Recall $e_n$ defined below \eqref{eq:initcont}. From \eqref{eq:heat1} and \eqref{eq:semi_disc} we get \vspace{-1mm}
\begin{equation} \label{eq:einit}
  \dot{e}_n(t) = P_n e_n(t) + \eta^1_n(t)+\eta^2_n(t)+\eta^3_n(t), \vspace{-1mm}
\end{equation}
where $\eta_n^1(t)= [P_n R_n \nu + B_n](f(t)-f_n(t))$, $\eta_2(t)=R_n \Pscr u(\cdot,t) - P_n R_n u(\cdot,t)-B_n f(t)$ and $\eta_n^3(t)= R_n \nu (\dot{f}_n(t) - \dot{f}(t))$.
From \eqref{eq:nulemma}, the estimate for $\|R_n\Pscr\nu\|_{2d}$ given below it and \eqref{eq:fn_limit} we get that $\lim_{n\to\infty} \sup_{t\in[0,T]} \|\eta_n^1(t)\|_{2d} =0$. Since $u\in C([\rho,T];PC^{(2),1} [0,1])$ for any given $\rho\in(0,T)$, see Lemma \ref{lm:reg}, applying Lemma \ref{lm:approxAq} with $\xi=u$ we get $\sup_{t\in[\rho,T]}\|\eta_n^2(t)\|_{2d} \leq  \sqrt{h}\tilde C(\rho) $ for all $n\gg1$ and a constant $\tilde C(\rho)$ independent of $n$. So $\lim_{n\to\infty}\sup_{t\in[\rho,T]} \|\eta_n^2(t)\|_{2d}=0$. Finally, it follows from \eqref{eq:fn_limit} and the estimate for $\|R_n\nu\|_{2d}$ below \eqref{eq:nulemma} that $\lim_{n\to\infty}\sup_{t\in[0,T]} \|\eta_n^3(t)\|_{2d}=0$. In summary we get that for any $\rho>0$, \vspace{-1mm}
\begin{equation} \label{eq:eta_conv}
 \lim_{n\to\infty}  \sup_{t\in [\rho,T]}\|\eta_n^1(t)+ \eta_n^2(t)+\eta_n^3(t)\|_{2d} =0. \vspace{-1mm}
\end{equation}
Applying Proposition \ref{pr:eigenA} to \eqref{eq:einit} and using \eqref{eq:eta_conv} we get \eqref{eq:enl2}.
\end{proof}

We will now show that, under further assumptions on the initial state $u_0$ and inputs $f$ and $f_n$ the solutions of the $n^{\rm th}$-order semi-discrete approximation \eqref{eq:semi_disc} converge pointwise to the solution of the PIDE \eqref{eq:heat1}-\eqref{eq:heat2} as $n\to\infty$.

\begin{corollary}\label{cor:semidisc_conv}
Suppose that in Theorem \ref{th:semidisc_conv} we let $u_0\in\Dscr(\Pscr^\infty)$, $f, f_n\in C^\infty[0,T]$ with $\lim_{n\to\infty}\|f_n-f\|_{C^k[0,T]}=0$ and  $f_n^{(k)}(0)=0$ for all integers $k\geq0$ and $n\geq1$. Then \vspace{-1mm}
\begin{equation} \label{eq:semdisconv8b}
  \lim_{n\to\infty}\sup_{t\in [0,T]}\|R_n u(\cdot,t) - v_n(t)\|_\infty = 0.\vspace{-1mm}
\end{equation}
Furthermore, we also have \vspace{-2mm}
\begin{equation} \label{eq:semdisconv8}
  \lim_{n\to\infty}\sup_{t\in[0,T]}|u(0,t) - v_{n,1}(t)| = 0.
\end{equation}
\end{corollary}
\begin{proof}
Note that $f^{(k)}(0)=0$ for all integers $k\geq0$. Using \cite[Ch. 4, Thm. 3.5 (ii)]{Pazy:1983} instead of \cite[Ch. 4, Thm. 3.5 (i)]{Pazy:1983} in the proof of Lemma \ref{lm:reg} it follows that the solution $u$ of \eqref{eq:heat1}-\eqref{eq:heat2} with initial state $u_0$ and input $f$ is in $C([0,T]; PC^{(2),1}[0,1])$. Furthermore, $\dot u \in C([0,T];C^1[0,1])$ and $u$ satisfies \eqref{eq:heat1}-\eqref{eq:heat2} for $t\geq0$. (The result \cite[Ch. 4, Thm. 3.5 (ii)]{Pazy:1983} enables us to replace $(0,T]$ in Lemma \ref{lm:reg} with $[0,T]$). Let $\bar u$ denote the solution of \eqref{eq:heat1}-\eqref{eq:heat2} with initial state $\Pscr u_0$ and input $\dot f$. Then, like $u$ above, we get $\bar u \in C([0,T]; PC^{(2),1}[0,1])\cap C^1([0,T]; C^1[0,1])$ satisfies \eqref{eq:heat1}, $\alpha_0 \bar u_x(1,t) + \beta_0 \bar u(1,t)=0$ and $\alpha_1 \bar u_x(1,t) + \beta_1 \bar u(1,t)=\dot f(t)$ for $t\in[0,T]$. Differentiating \eqref{eq:mildsoln} it is easy to check that the resulting formula for $\dot u$ is in fact the convolution formula for $\bar u$, i.e. $\bar u =\dot u$ and $\dot u$ has the same regularity as $\bar u$.

Recall $e_n(t)$ defined below \eqref{eq:initcont} and $\eta_n^1(t)$, $\eta_n^2(t)$, $\eta_n^3(t)$ introduced in \eqref{eq:einit}. Let $\eta_n(t) \!=\!\eta_n^1(t)+\eta_n^2(t)+\eta_n^3(t)$. Clearly $e_n(0)=0$ and $\dot e_n(0)=R_n \Pscr u_0 - P_n R_n u_0$. Since $u\in C([0,T];PC^{(2),1}[0,1])$, we can take $\rho=0$ in the proof of Theorem \ref{th:semidisc_conv} to conclude from \eqref{eq:enl2} that \vspace{-1.5mm}
\begin{equation} \label{eq:en_rho0}
  \lim_{n\to\infty}\sup_{t\in [0,T]} \|e_n(t)\|_{2d} =0. \vspace{-1.5mm}
\end{equation}
Using $\dot u\in C([0,T];PC^{(2),1}[0,1])$ and $\alpha_1 \dot u_x(1,t) + \beta_1 \dot u(1,t)=\dot f(t)$, we can replicate the arguments used to derive \eqref{eq:eta_conv} to conclude that $\lim_{n\to\infty} \sup_{t\in [0,T]}\|\dot\eta_n(t)\|_{2d} =0$. Consequently, taking the derivative of \eqref{eq:einit} and applying Proposition \ref{pr:eigenA} to the resulting equation, and inferring using Lemma \ref{lm:approxAq} that $\lim_{n\to\infty}\|\dot e_n(0)\|_{2d}=0$, we get \vspace{-1.4mm}
\begin{equation} \label{eq:doten_rho0}
 \lim_{n\to\infty}\sup_{t\in [0,T]} \|\dot e_n(t)\|_{2d} =0. \vspace{-1.6mm}
\end{equation}
Using Lemma \eqref{lm:approxAq} it is easy to mimic the steps used to derive \eqref{eq:eta_conv} to show that $\lim_{n\to\infty} \sup_{t\in[0,T]} \|\eta_n(t)\|_{2d}=0$. From this, \eqref{eq:doten_rho0}, \eqref{eq:einit} and $e_n, \dot e_n, \eta_n\in C([0,T];\rline^n)$ we get \vspace{-1.4mm}
\begin{equation}\label{eq:estAn3}
  \lim_{n\to\infty}\sup_{t\in [0,T]} \|P_n e_n(t)\|_{2d} =0. \vspace{-1.6mm}
\end{equation}
Using \eqref{eq:en_rho0} and \eqref{eq:estAn3}, and appealing to Lemma \ref{lm:disc_sob}, we get \vspace{-1.4mm}
\begin{equation}\label{eq:estDen8}
 \lim_{n\to\infty}\sup_{t\in [0,T]}\| e_n(t) \|_\infty = 0. \vspace{-1.6mm}
\end{equation}
for all $n\gg1$ and $t\in[0,T]$. Recall that $ R_n u(\cdot,t) - v_n(t) = e_n(t) - R_n \nu (f_n(t)-f(t))$. We obtain  \eqref{eq:semdisconv8b} from this, by bounding $e_n(t)$ using \eqref{eq:estDen8}, and bounding $R_n \nu (f_n(t)-f(t))$ using \eqref{eq:fn_limit} and $\|R_n \nu\|_{\infty}\leq \mu_1$ (see the definition of $\nu$ below \eqref{eq:mildsoln}). Observe that from the definition of $e_n$ we \vspace{-2mm} have
\begin{align}
 u(0,t) - v_{n,1}(t) &\,= u(0,t) - u(h,t) + e_{n,1}(t)\nonumber\\
  &\hspace{10mm} + \nu(h) (f(t)-f_n(t)), \label{eq:inff1}\\[-4.6ex]\nonumber
\end{align}
Since $u\in C([0,T];PC^{(2),1}[0,1])$ we get $\lim_{h\to0}|u(0,t)-u(h,t)|\!=\!0$ and from \eqref{eq:estDen8} we get $\lim_{n\to\infty}e_{n,1}(t)=0$, uniformly for $t\!\in\! [0,\!T]$. From the definition of $\nu$ we get $\lim_{h\to0}\! \nu(h)\!=\!0$. Using these limits in \eqref{eq:inff1} gives us \eqref{eq:semdisconv8}.
\end{proof}
\section{Numerical examples} \label{sec5} 

Consider the parabolic PIDE \eqref{eq:heat1}-\eqref{eq:heat2} with the following coefficients: $\theta(x)=1+x$ for $x\in[0,0.5)$, $\theta(x) = 2$ for $x\in [0.5,1]$, $\sigma(x)=2-2x$ for $x\in [0,0.3)$, $\sigma(x)=\sin(5\pi x)$ for $x\in [0.3,1]$, $\lambda(x)=e^{-5x}$ for $x\in [0,0.7)$, $\lambda(x)=2 x^4$ for $x\in [0.7,1]$, $\varphi(x,\tilde x)=1$ for all $x,\tilde x\in[0,1]$,  $\alpha_0=\beta_1 =1$ and $\beta_0=\alpha_1=0$. We present two numerical examples to illustrate our convergence results.

{\noindent\bf Example 1.} Fix $T=1$. The initial state of \eqref{eq:heat1}-\eqref{eq:heat2} is $u_0(x) = 0$ for $x\in [0,0.3] \cup [0.7,1]$ and $u_0(x)=0.5$ for $x\in (0.3,0.7)$. The input to the $n^{\rm th}$-order semi-discrete approximation \eqref{eq:semi_disc} is $f_n(t) = (1-1/n)e^{-t}\sin(\pi t)$ for $t\in [0,1]$ which converges to the input $f(t) = e^{-t}\sin(\pi t)$. Note that while Theorem \ref{th:semidisc_conv} is stated for $u_0\in C[0,1]$, the result remains true even when $u_0\in PC[0,1]$ as demonstrated here. The error $e_n^1 = \sup_{t\in [0,1]}\|S_{200}v_{200}(t)-S_n v_n(t)\|_{L^2(0,1)}$ for $n$ varying from 10 to 100 is shown in Figure \ref{fig1}. As expected from Theorem \ref{th:semidisc_conv}, the error $e_n^1$ converges to zero as $n$ approaches $\infty$. The solutions $S_nv_n$ for $n\in \{10,50,100\}$ are shown in Figure \ref{fig2}. \vspace{1mm}

\begin{figure}
\centerline{\includegraphics[scale=0.85]{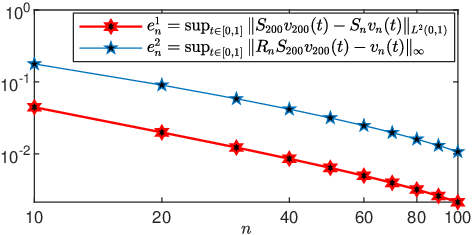}\vspace{-3mm}}
\caption{The errors $e_n^1$ and $e_n^2$ corresponding to the Examples 1 and 2 for $n\in \{10,20,\ldots 100\}$. The errors converge to zero as $n\to\infty$.}\label{fig1}\vspace{-4mm}
\end{figure}

{\noindent \bf Example 2.} Fix $T=1$. The initial for \eqref{eq:heat1}-\eqref{eq:heat2} is $u_0=0$. The input to the $n^{\rm th}$-order semi-discrete approximation \eqref{eq:semi_disc} is $f_n(t) \!=\! (1-1/n)e^{-(5t-5t^2)^{-2}}$ for $t\in (0,1)$ and $f_n(0)=f_n(1)=0$. Clearly, the limit $f$ of $f_n$ is given by $f(t)=  e^{-(5t-5t^2)^{-2}}$ for $t\in (0,1)$ and we have $f_n^{(k)}(0)=f^{(k)}(0)=0$ for all $k\geq 0$. So, the assumptions in Corollary \ref{cor:semidisc_conv} are satisfied. From Figure \ref{fig1} we can see that the error $e_n^2(t) = \sup_{t\in [0,1]}\|R_{200}S_{200} v_{200}(t)- v_n(t)\|_\infty$ converges to zero as $n$ approaches $\infty$ as expected from Corollary \ref{cor:semidisc_conv}. The solutions $S_n v_n$ for $n\in \{10,50,100\}$ are shown in Figure \ref{fig3}.

\begin{figure}
\centerline{\includegraphics[scale=0.55]{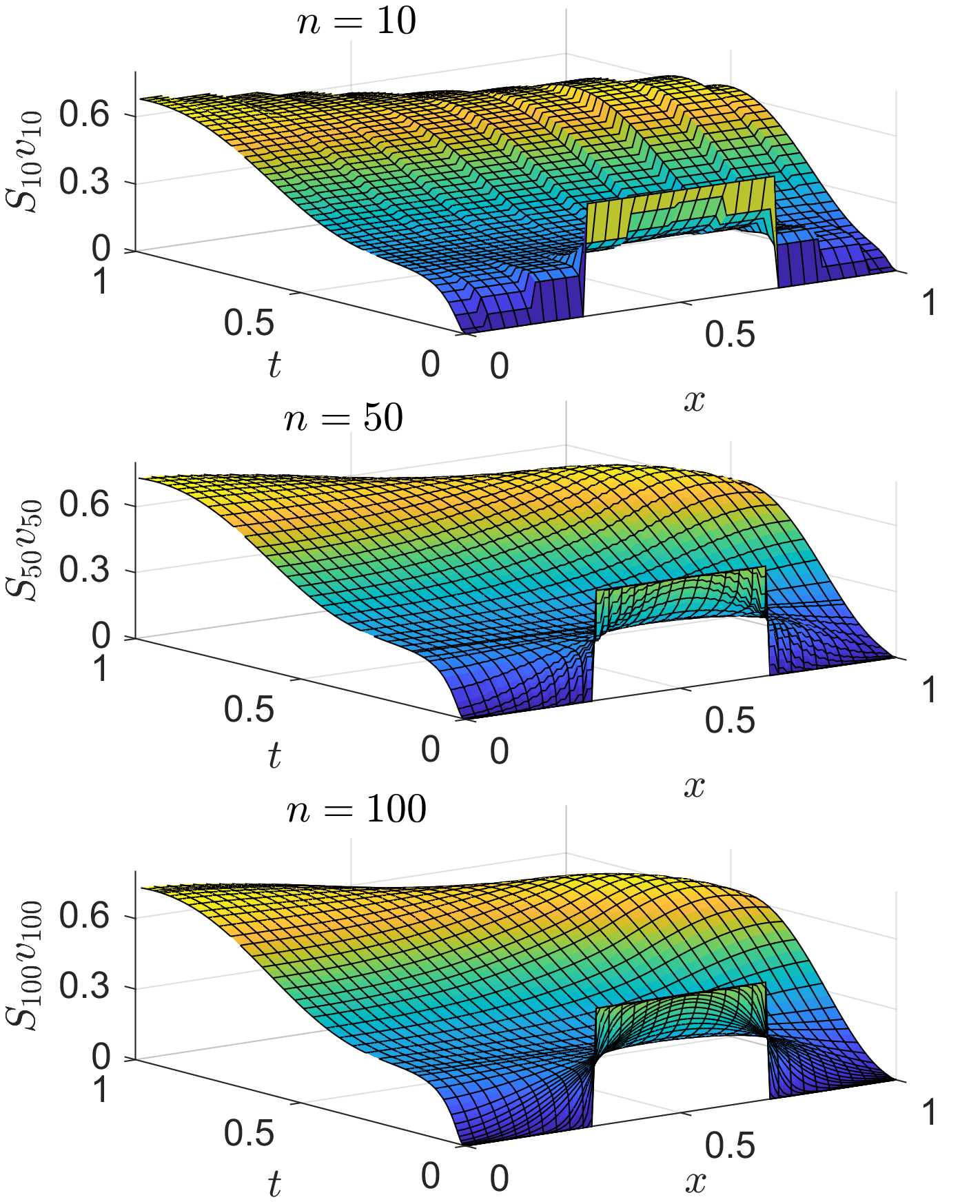}\vspace{-2.5mm}}
\caption{The solution $S_n v_n$ for $n\in \{10,50,100\}$ of the $n^{\rm th}$-order semi-discrete ODE \eqref{eq:semi_disc} for the initial state and input specified in  Example 1.}\label{fig2}\vspace{-6.5mm}
\end{figure}

\begin{figure}
\centerline{\includegraphics[scale=0.55]{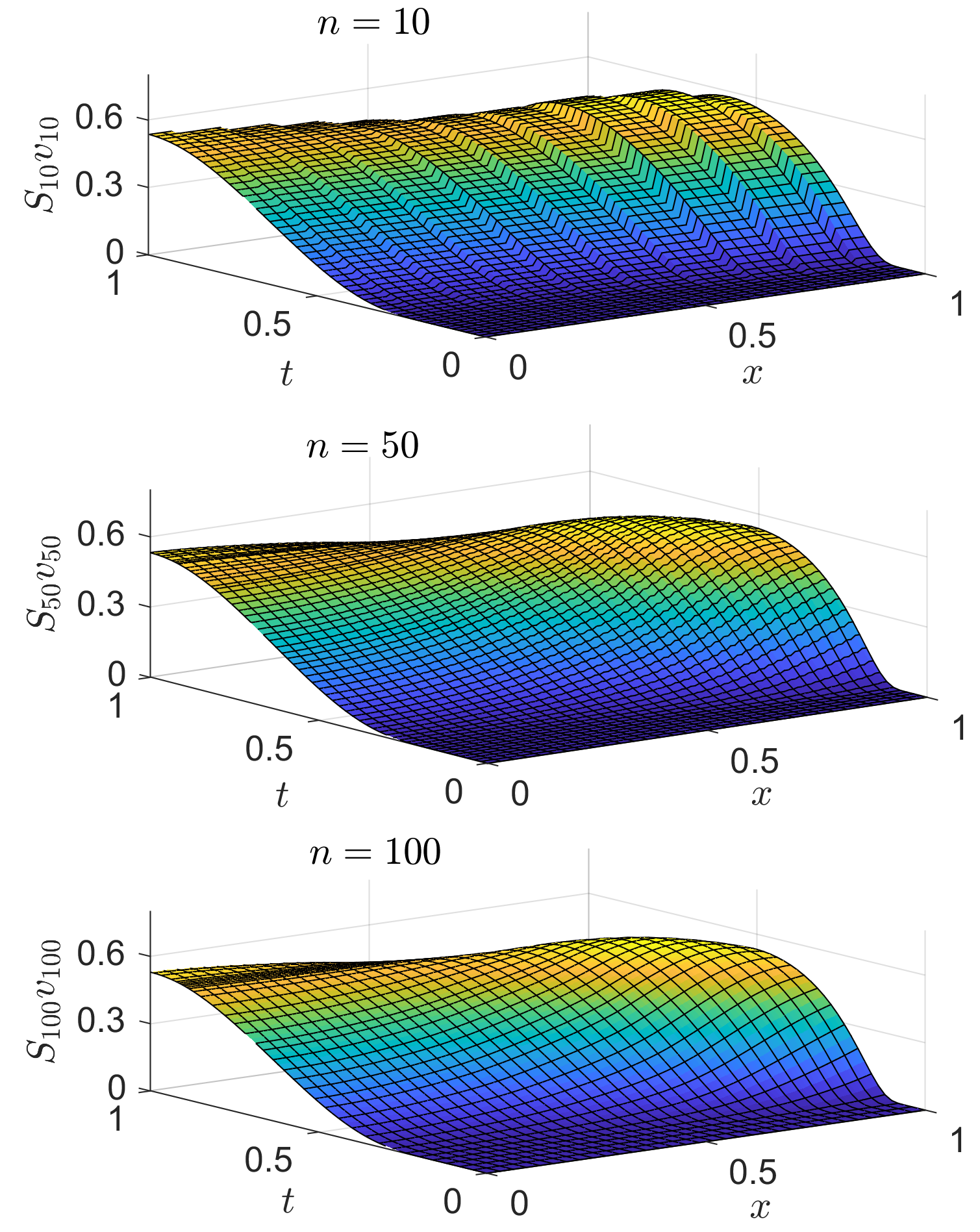}\vspace{-2.5mm}}
\caption{The solution $S_n v_n$ for $n\in \{10,50,100\}$ of the $n^{\rm th}$-order semi-discrete ODE \eqref{eq:semi_disc} for the initial state and input specified in Example 2.}\label{fig3} \vspace{-6.5mm}
\end{figure}


\end{document}